\documentclass[11pt]{article}
\usepackage[utf8]{inputenc}
\usepackage[english]{babel}
\usepackage{amsmath,xcolor}
\usepackage{amsfonts}
\usepackage[margin=0.75in]{geometry}
\usepackage{listings}
\usepackage{graphicx}
\usepackage{float}

\usepackage{csquotes}

\usepackage{mathtools}

\usepackage{amsthm}
\usepackage{cleveref}
\usepackage{lipsum}

\usepackage{verbatim}

\newtheorem{definition}{Definition}[section]

\newtheorem{theorem}{Theorem}[section]
\newtheorem{corollary}{Corollary}[theorem]
\newtheorem{lemma}[theorem]{Lemma}
\newtheorem{proposition}[theorem]{Proposition}
\newtheorem{informal}[theorem]{Informal Theorem}

\newcommand{\snm}[3]{\ensuremath{#1\text{-SNM}_{#2}\text{-}#3}}

\usepackage{amssymb}

\title{Approximately Strategyproof Tournament Rules in the Probabilistic Setting}
\author{Kimberly Ding and S. Matthew Weinberg}
\date{}

\begin{document}

\maketitle
\addtocounter{page}{-1}
\begin{abstract}
We consider the manipulability of tournament rules which map the results of $\binom{n}{2}$ pairwise matches and select a winner. Prior work designs simple tournament rules such that no pair of teams can manipulate the outcome of their match to improve their probability of winning by more than $1/3$, and this is the best possible among any Condorcet-consistent tournament rule (which selects an undefeated team whenever one exists)~\cite{SchneiderSW17,SchvartzmanWZZ20}. These lower bounds require the manipulators to know precisely the outcome of all future matches.

We take a beyond worst-case view and instead consider tournaments which are ``close to uniform'': the outcome of all matches are independent, and no team is believed to win any match with probability exceeding $1/2+\varepsilon$. We show that Randomized Single Elimination Bracket~\cite{SchneiderSW17} and a new tournament rule we term Randomized Death Match have the property that no pair of teams can manipulate the outcome of their match to improve their probability of winning by more than $\varepsilon/3 + 2\varepsilon^2/3$, for all $\varepsilon$, and this is the best possible among any Condorcet-consistent tournament rule.

Our main technical contribution is a recursive framework to analyze the manipulability of certain forms of tournament rules. In addition to our main results, this view helps streamline previous analysis of Randomized Single Elimination Bracket, and may be of independent interest.
\end{abstract}

\newpage
\section{Introduction}
A tournament consists of $n$ teams competing to win a championship via pairwise matches, and a tournament rule (possibly randomly) selects a single winner as a result of these matches. Tournament rules have been studied within Social Choice Theory for decades~\cite{Fishburn77,Miller80,ShepsleW84, Banks85, FisherR92,LaffondLB93,
Laslier97}, see also~\cite{BrandtCELP16} for a survey, and have gained attention from a few angles within TCS more recently~\cite{AltmanPT09, AltmanK10,Williams10, StantonW11, KimW15,KimSW16, SchneiderSW17, SchvartzmanWZZ20}. Our work follows the model studied in~\cite{AltmanPT09, AltmanK10, SchneiderSW17, SchvartzmanWZZ20} and seeks to design tournaments which are both fair (in that they select a reasonable winner, based on the match outcomes), and ``as strategyproof as possible'' subject to this.

More specifically, these works acknowledge that an undefeated team, if one exists, should surely win any reasonable tournament format. Formally, this property is termed \emph{Condorcet-consistent} (Definition~\ref{def:cc}). These works also consider the possibility of two teams strategically manipulating the match between them to improve the probability that one of them wins. The situation to have in mind is that perhaps two teams sponsored by the same company enter an eSports tournament, and wish to maximize the probability that either of them take home the prize money. A rule is $2$-Strongly Non-Manipulable ($2$-SNM, Definition~\ref{def:snm}) if no pair of teams can successfully manipulate it to improve the probability that one of them wins.

Initial works quickly established that no tournament rule exists which is both Condorcet-consistent and $2$-SNM~\cite{AltmanPT09, AltmanK10}. More recent works study the extent to which tournament rules can be Condorcet-consistent and \emph{approximately} $2$-SNM. Specifically, a rule is $2$-SNM-$\alpha$ if no pair of teams can improve the probability that one of them wins by more than $\alpha$ (Definition~\ref{def:snm}).~\cite{SchneiderSW17,SchvartzmanWZZ20} design simple Condorcet-consistent tournament rules (which we define in Section~\ref{sec:recursive}) which are $2$-SNM-$1/3$, and also establish that this is the best possible \emph{worst-case} guarantee.

\subsection{Our Results: Probabilistic Tournaments}
The lower bounds in previous works assume that the \emph{deterministic future outcome} of all matches is known at the time of manipulation. While this is certainly a plausible scenario, most competitions worth watching have some element of uncertainty, even for matches between strong and weak teams. Indeed,~\cite[Open Problem~2]{SchneiderSW17} explicitly asks whether improved guarantees are possible if instead the teams have a common Bayesian prior about the possible outcomes of future matches which is \emph{bounded away from determinstic}. 

For example, consider the case where the outcome of all matches are uniformly at random. Then it is not hard to design a Condorcet-consistent tournament rule which is $2$-SNM in this case. One example is a simple single-elimination bracket: when two manipulating teams face each other, each of them is equally likely to continue on and win the tournament, so manipulating doesn't help. 

What if instead the outcome of all matches are not uniformly random, but close? More specifically, what if the match results are independent, and no team wins any match with probability more than $1/2+\varepsilon$? When $\varepsilon =0$, the previous paragraph establishes that rules exist where profitable manipulation is impossible. When $\varepsilon = 1/2$,~\cite{SchneiderSW17,SchvartzmanWZZ20} establish that $2$-SNM-$1/3$ tournaments exist, but no better. What about when $\varepsilon \in (0,1/2)$? How do the achievable guarantees vary as a function of $\varepsilon$?

Our main result resolves this question, and nails down the guarantees precisely as a function of $\varepsilon$. Moreover, we show that \emph{the same tournament rule achieves the optimal guarantee for all $\varepsilon$}. Below, Randomized Single Elimination Bracket (Definition~\ref{def:rseb}, henceforth RSEB) randomly seeds all teams, then runs a single-elimination bracket to determine the winner, and was shown to be $2$-SNM-$1/3$ in~\cite{SchneiderSW17}. Randomized Death Match (Definition~\ref{def:rdm}, henceforth RDM) repeatedly picks two uniformly random teams to play a match, and eliminates the loser (and is first analyzed in this paper).

\begin{informal}[See Theorems~\ref{thm:RDMmain},~\ref{thm:RSEBmain}] For all $\varepsilon \in [0,1/2]$, Randomized Death Match and Randomized Single Elimination Bracket are $2$-SNM-$(\varepsilon/3+2\varepsilon^2/3)$ when match outcomes are independent, and no team wins any match with probability more than $1/2+\varepsilon$. Moreover, for all $\varepsilon \in [0,1/2]$, no Condorcet-consistent tournament rule is $2$-SNM-$\alpha$ for any $\alpha < \varepsilon/3+2\varepsilon^2/3$, on the set of all tournaments with independent match outcomes that no team wins with probability more than $1/2+\varepsilon$.
\end{informal}

\subsection{Technical Highlights}
We prove our main result by finding a strong structural similarity between tournaments like RDM and RSEB: they can be defined recursively. Specifically, RDM could be alternatively defined as ``Pick two teams uniformly at random, and eliminate the loser of their match. Then, recurse on the remaining teams.'' Similarly, RSEB can be alternatively defined as ``Pick a uniformly random perfect matching between the teams, and eliminate all teams which lose their match. Then, recurse on the remaining teams.'' Even the Randomized King of the Hill rule (Definition~\ref{def:rkoth}, henceforth RKotH) defined by~\cite{SchvartzmanWZZ20} fits this framework as well: ``Pick a uniformly random team to play all other teams. Eliminate all teams who lose a match, and recurse on the remaining teams.'' 

In Section~\ref{sec:recursive}, we give a formal definition of what it means to be a recursive tournament rule. And in Section~\ref{sec:main} (specifically, Theorem~\ref{thm:recursive}), we provide a general framework to analyze the manipulability of recursive tournament rules on probabilistic tournaments. This provides a fairly clean outline to analyze recursive tournament rules, and our main result then applies this framework to RDM and RSEB. In the $\varepsilon = 1/2$ case, our analysis of RSEB in isolation is perhaps not much simpler than that of~\cite{SchneiderSW17}, but our proof is arguably more structured. Indeed, a substantial fraction of our proof can be applied verbatim to other tournament rules like RDM, or applied verbatim to the $\varepsilon < 1/2$ case.

It is worth noting that our framework does face some technical barriers in accommodating RKotH (and we leave open whether RKotH achieves the same guarantees as RDM and RSEB). But, the technical barrier is easy to describe: the matches played in each round of RDM and RSEB form a matching --- no team plays more than one match. In RKotH, some team plays multiple matches. It seems likely that our analysis would extend (perhaps with messier calculations) to any recursive rule where each round's matches form a matching. But we highlight the aspects of our analysis which rely on this aspect of RDM/RSEB (and therefore don't hold for RKotH), and believe this is a genuine barrier. 

\subsection{Further Related Work}
We've already discussed the most related work~\cite{AltmanPT09,AltmanK10,SchneiderSW17, SchvartzmanWZZ20}. The model is first posed in~\cite{AltmanPT09}, and~\cite{AltmanK10} design tournaments which are $2$-SNM and approximately Condorcet-consistent (e.g. pick a uniformly random match and declare the winner of that match to win the tournament).~\cite{SchneiderSW17} first proposed to instead consider rules which are Condorcet-consistent and approximately strategyproof, and establishes that RSEB is $2$-SNM-$1/3$ and that this is optimal.~\cite{SchvartzmanWZZ20} considers larger manipulating sets (not relevant to this paper) and also designs RKotH, showing that it too is $2$-SNM-$1/3$ and satisfies a stronger notion of fairness termed ``cover-consistent''. In relation to these works, our main contribution is going beyond the worst-case to derive improved bounds when match outcomes are more uncertain. A technical contribution is our framework of recursive tournament rules.

Other recent works within TCS focus specifically on single-elimination brackets and manipulation in the form of a bracket designer trying to get a certain team to win~\cite{Williams10,StantonW11,KimW15,KimSW16}, or manipulability of particular tournament rules such as the World Cup qualifying procedure~\cite{Pauly14,Csato17}. Aside from being thematically related, there is no significant technical overlap with our work.

\subsection{Roadmap}
Section~\ref{sec:prelim} immediately follows with definitions and preliminaries. Section~\ref{sec:recursive} provides definitions concerning recursive tournament rules, and formally defines RDM and RSEB. Section~\ref{sec:main} provides our framework for analyzing the manipulability of recursive tournament rules. Section~\ref{sec:deterministic} applies this framework, as a warmup, to rederive the main result of~\cite{SchneiderSW17} and analyze RDM/RSEB in the deterministic case. Section~\ref{sec:epsilon} proves our main results, and Section~\ref{sec:conclusion} provides a brief conclusion.

\section{Preliminaries}\label{sec:prelim}
\subsection{Tournament Rule Basics}
In this section, we introduce notation consistent with prior work~\cite{AltmanK10, SchneiderSW17, SchvartzmanWZZ20}. 

\begin{definition}[Deterministic Tournament] A (round robin) tournament $T$ on $n$ teams is a complete, directed graph on $n$ vertices whose edges denote the outcome of a match between two teams. Team $i$ beats team $j$ if the edge between them points from $i$ to $j$. 
\end{definition}

\begin{definition}[Tournament Rule] A tournament rule $r$ is a function that maps (deterministic) tournaments $T$ to a distribution over teams, where $r_i(T):=\Pr[r(T)=i]$ denotes the probability that team $i$ is declared the winner of tournament $T$ under rule $r$. We use the shorthand $r_S(T):=\sum_{i \in S} r_i(T)$ to denote the probability that a team in $S$ is declared the winner of tournament $T$ under rule $r$.
\end{definition}

Finally, we are interested in tournament rules which satisfy basic notions of fairness. Importantly, note that Condorcet-consistence is a minimal notion of fairness, and in particular does not constrain the behavior of $r$ on any tournament without a Condorcet winner.

\begin{definition}[Condorcet-Consistent]\label{def:cc} Team $i$ is a Condorcet winner of a tournament $T$ if $i$ beats every other team (under $T$). A tournament rule $r$ is Condorcet-consistent if for every tournament $T$ with a Condorcet winner $i$, $r_i(T) = 1$ (whenever $T$ has a Condorcet winner, that team wins with probability $1$).
\end{definition}

\subsection{Independent Probabilistic Tournaments}
In this work, we study probabilistic tournaments. That is, we are interested in tournaments where the outcome of each match is not known to teams ``in advance'', but teams share a Bayesian prior about the likelihood of each possible outcome. In particular, we consider when match outcomes are independent.

\begin{definition}[Independent Probabilistic Tournament]\label{def:independent} A \emph{probabilistic tournament} $T$ is just a distribution over deterministic tournaments. For notational convenience, we slightly abuse notation and refer by $r_i(T)$ to $\mathbb{E}[r_i(T)]$ (that is, $r_i(T)$ is the probability that team $i$ wins when rule $r$ is applied to $T$, over the randomness in $r$ \emph{and the randomness in drawing $T$}). A probabilistic tournament $T$ is \emph{independent} if all match outcomes in $T$ are independent events. Observe that a probabilistic tournament $T$ is fully defined by probabilities $p^T_{ij}$ for all $i < j$, where $p^T_{ij}$ denotes the probability that $i$ beats $j$ in tournament $T$. 
\end{definition}

Observe that deterministic tournaments are also independent probabilistic tournaments, with each $p_{ij}^T \in \{0,1\}$. Like prior work, we study tournament rules which are ``as strategyproof as possible''. Because of our focus on independent probabilistic tournaments, we first refine previous definitions of non-manipulability.

\begin{definition}[$S$-adjacent] Two independent probabilistic tournaments $T, T'$ are \emph{$S$-adjacent} if $p^T_{ij} = p^{T'}_{ij}$ whenever $\{i,j\} \not \subseteq S$. That is, two independent probabilistic tournaments are $S$-adjacent when all (probabilistic) match outcomes are identical, except possibly for matches between two teams in $S$.
\end{definition}

Intuitively, two tournaments $T, T'$ are $S$-adjacent if the set of teams $S$ can manipulate the outcomes of matches between them in advance and cause the resulting (probabilistic) tournament to go from $T$ to $T'$.

\begin{definition}[Manipulating a Tournament]\label{def:snm} For a set $S$ of teams, independent probabilistic tournament $T$, and tournament rule $r$, we define $\alpha^r_S(T)$ to be the maximum winning probability that $S$ can possibly gain by manipulating $T$ to an $S$-adjacent $T'$. That is: $\alpha^r_S(T):=\max_{T': T' \text{ is $S$-adjacent to } T}\{r_S(T') - r_S(T)\}$.

For a class $\mathcal{T}$ of (independent probabilistic) tournaments, we also define $\alpha_k^r(\mathcal{T}):=\max_{T \in \mathcal{T}, S:|S| \leq k}\{\alpha_S^r(T)\}$. If $\alpha_k^r(\mathcal{T}) \leq \alpha$, we say $r$ is \emph{$k$-Strongly Non-Manipulable at probability $\alpha$ with respect to $\mathcal{T}$} (\snm{k}{\mathcal{T}}{\alpha}). To match notation of prior work, we also say a tournament is $k$-SNM-$\alpha$ if it is \snm{k}{\mathcal{S}^{1/2}}{\alpha}.\footnote{Note that $\mathcal{S}^{1/2}$ is the set of all deterministic tournaments, defined shortly below.}

Finally, we also define $\alpha_k(\mathcal{T}) = \inf_{\text{Condorcet consistent } r}\{\alpha_k^r(\mathcal{T})\}$.
\end{definition}

Intuitively, $r$ is \snm{k}{\mathcal{T}}{\alpha} if no colluding set of $\leq k$ teams can manipulate a tournament in $\mathcal{T}$ to improve the probability the winner is in $S$ by more than $\alpha$. The refinement over prior work is that the condition only holds for tournaments in $\mathcal{T}$ --- prior work only considers guarantees that hold over \emph{all} tournaments. The additional notation in Definition~\ref{def:snm} are just terms that will be helpful for later exposition.

Finally, we focus on independent probabilistic tournaments that are close to uniformly at random. 

\begin{definition}[$\varepsilon$-Bounded Tournaments] An independent probabilistic tournament $T$ is \emph{weakly $\varepsilon$-bounded} if for all $i,j$ $p^T_{ij} \in [1/2-\varepsilon,1/2+\varepsilon]$. We refer to $\mathcal{T}^\varepsilon$ as the set of all $\varepsilon$-bounded tournaments.

An independent probabilistic tournament $T$ is \emph{strictly $\varepsilon$-bounded} if for all $i,j$ $p^T_{ij} \in \{1/2-\varepsilon,1/2+\varepsilon\}$. We refer to $\mathcal{S}^\varepsilon$ as the set of all strictly $\varepsilon$-bounded tournaments. It will be helpful to define the notation $\mathcal{T}^\varepsilon_{\leq n}$ (respectively, $\mathcal{S}^\varepsilon_{\leq n}$) as the set of all weakly- (respectively, strictly-) bounded tournaments on $\leq n$ teams.
\end{definition}

For example, every deterministic tournament is $1/2$-bounded, and the uniformly random tournament is $0$-bounded. Our main results study $\alpha_k(\mathcal{T}^\varepsilon)$ as a function of $\varepsilon$. We conclude with a brief lemma relating $\alpha_k(\mathcal{T}^\varepsilon)$ to $\alpha_k(\mathcal{S}^\varepsilon)$, as direct analysis of $\alpha_k(\mathcal{S}^\varepsilon)$ is significantly simpler than direct analysis of $\alpha_k(\mathcal{T}^\varepsilon)$.

\begin{proposition}\label{prop:convex} For all rules $r$, and all $\varepsilon, k$, $\alpha_k^r(\mathcal{T}^\varepsilon) = \alpha_k^r(\mathcal{S}^\varepsilon)$. Therefore, for all $\varepsilon, k$, $\alpha_k(\mathcal{T}^\varepsilon) = \alpha_k(\mathcal{S}^\varepsilon)$.
\end{proposition}

While the proof is deferred to Appendix~\ref{app:prelim}, the high-level outline is fairly intuitive. First, the ``Therefore,\ldots'' statement follows trivially from the first portion of the lemma. Also, it is trivial to see that $\alpha_k^r(\mathcal{T}^\varepsilon) \geq \alpha_k^r(\mathcal{S}^\varepsilon)$, as $\mathcal{S}^\varepsilon \subseteq \mathcal{T}^\varepsilon$. So the interesting step is establishing $\alpha_k^r(\mathcal{T}^\varepsilon) \leq \alpha_k^r(\mathcal{S}^\varepsilon)$. Intuitively, this follows because all (independent probabilistic) tournaments in $\mathcal{T}^\varepsilon$ can be written as convex combinations of tournaments in $\mathcal{S}^\varepsilon$, and one might expect that any particular tournament rule is most manipulable on extreme points (indeed this is true). 
With Proposition~\ref{prop:convex}, we may restrict our study to $\alpha_k(\mathcal{S}^\varepsilon)$.

\section{Recursive Tournament Rules}\label{sec:recursive}
In this section, we formalize a class of tournament rules which have a recursive aspect to them. This will help us streamline previous analysis of~\cite{SchneiderSW17}, and also easily design a new tournament rule with matching guarantees. In addition, this view will help give us a clean outline to analyze the performance of these rules on tournaments in $\mathcal{T}^\varepsilon$, rather than just in the worst case.

\begin{definition}[Elimination Rule] An \emph{elimination rule} $E$ takes as input a number $n$ of teams and selects (possibly randomly) a set $M:=E(n)$ of matches to play, with $|M| \in [1,n-1]$.

An elimination rule is \emph{matching} if $M$ is a (not necessarily perfect) matching with probability one.
\end{definition}

\begin{definition}[Recursive Tournament Rule] A \emph{recursive tournament rule} is fully defined by its elimination rule $E$. The recursive tournament rule $r^E$ takes as input a tournament $T$, samples matches $M:=E(n)$ to play, and then \emph{eliminates} any team which loses a match in $M$.\footnote{Observe that because $|M| \in [1,n-1]$ that at least one team is eliminated, but not all teams are eliminated.} Specifically, $T|_M$ denotes the induced subgraph of $T$ on teams not eliminated by the matches in $M$. The tournament then recursively executes $r^E(T|_M)$ to select a winner. As a base case, when there is only one team left, that team is the winner.
\end{definition}

We now give three examples of elimination rules, and the resulting tournament rule. Randomized Single Elimination Bracket was first studied in~\cite{SchneiderSW17}, Randomized King of the Hill was first studied in~\cite{SchvartzmanWZZ20}, and Randomized Death Match is first studied in this paper. 

\begin{definition}[Randomized Single Elimination Bracket]\label{def:rseb} For a tournament $T$ on $n$ teams, let $n':=2^{\lceil\log_2 n\rceil}$. Create $n'-n$ dummy players who all lose to the original $n$ teams. Let $M$ be matches corresponding to a uniformly random perfect matching (i.e. exactly $n'/2$ matches are played, and every team plays in exactly one match). Eliminate the losers and recurse on the remaining (non-dummy) teams. 
\end{definition}

\begin{definition}[Randomized King of the Hill]\label{def:rkoth} Pick a uniformly random team $i$, and have every team play $i$. Observe that if $i$ is a Condorcet winner, then $i$ will be the only remaining team. Otherwise, $i$ and every team it beats will be eliminated. Recurse on the remaining teams.
\end{definition}

\begin{definition}[Randomized Death Match]\label{def:rdm} Pick two uniformly random teams (without replacement) and play their match. Eliminate the loser and recurse on the remaining teams.
\end{definition}

Observe that our definition of Randomized Single Elimination Bracket (RSEB) differs semantically from that given in~\cite{SchneiderSW17}, where the $n'$ teams are uniformly permuted into $n'$ seeds, and then the resulting bracket is played (without re-randomizing at each round). Observe that the two definitions are equivalent (identically distributed), however, as our definition simply produces the seeding by first figuring out the first-round matches, then the second-round matches, etc. Our definition of Randomized King of the Hill (RKotH) is identical (semantically) to that given in~\cite{SchvartzmanWZZ20}. 

Randomized Death Match (RDM) is similar to Randomized Voting Caterpillar (RVC)~\cite{SchneiderSW17}. Like RDM, RVC picks two uniformly random teams (without replacement) and eliminates the loser. However, rather than a ``pure recursion'', RVC proceeds by picking \emph{one} uniformly random remaining team to play the previous winner. This subtle distinction causes RDM to be $2$-SNM-$1/3$ (Theorem~\ref{thm:RDMeasy}), but not RVC (\cite[Theorem 3.14]{SchneiderSW17}). Observe also that RSEB and RDM have matching elimination rules, but RKotH doesn't (this is the main technical barrier in extending our analysis to RKotH).

Finally, observe that all three rules above are \emph{anonymous}: relabeling the teams simply relabels the distribution over winners. This will play a role in our analysis.

\begin{definition}[Anonymous] A tournament rule $r$ is \emph{anonymous} if for every tournament $T$, and every permutation $\sigma$, and all $i$, $r_{\sigma(i)}(\sigma(T))=r_i(T)$.
\end{definition}

\section{Key Framework}\label{sec:main}
In this section, we propose an outline to analyze the manipulability of recursive tournament rules for independent probabilistic tournaments. Below, for notational convenience, if a team $i$ is not present in tournament $T$ (e.g. because they were eliminated in an earlier round), we abuse notation and denote by $r_i(T):=0$. Additionally, we will let $M(T)$ denote the outcome of matches $M$ for tournament $T$ (i.e. for all $(u,v) \in M$, whether $u$ or $v$ wins in tournament $T$). Note that if $T$ is probabilistic, then $M(T)$ is a random variable, even after conditioning on $M$. 

The first step in our framework simply observes that recursive tournament rules can be analyzed recursively due to linearity of expectation. A (short) proof of Lemma~\ref{lem:recurse} appears in Appendix~\ref{app:recursive}.

\begin{lemma}\label{lem:recurse} Let $r^E$ be any recursive tournament rule. Let $T, T'$ be independent probabilistic tournaments on $n > 1$ teams, and let $S$ be any subset of teams. Then:
$$r^E_S(T') - r^E_S(T) = \mathbb{E}_{M\leftarrow E(n)}[r^E_S(T'|_M) - r^E_S(T|_M)].$$
\end{lemma}
To help parse notation: recall that if $T, T'$ are not deterministic, then the notation $r_S^E(T')$ (respectively, $r_S^E(T), r_S^E(T'|_M), r_S^E(T|_M)$) is taking an expectation over $T'$ (respectively, $T, T'|_M, T|_M$), as per Definition~\ref{def:independent}.\footnote{In particular, note that $T|_M$ and $T'|_M$ are not independent probabilistic tournaments, but are still probabilistic tournaments, as they are distributions over independent probabilistic tournaments.} On the right-hand side, we are taking an expectation first over the matches $M$ which are played. Inside the expectation, the teams in $T|_M$ and $T'|_M$ are still random variables (because they depend on the \emph{outcome} of matches in $M$ in $T$, which are still random after conditioning on $M$). And after the teams are determined by $M(T),M(T')$, the tournaments $T'|_M, T|_M$ are still probabilistic tournaments.

Importantly, observe that when $S = \{u,v\}$ and $T, T'$ are $S$-adjacent, the tournaments $T|_M$ and $T'|_M$ may differ for two reasons. First, perhaps $(u,v) \in M$. In this case, perhaps $M(T) \neq M(T')$ (because $T$ and $T'$ can differ on the $(u,v)$ match), and then $T|_M,T'|_M$ may have different sets of teams. Second, perhaps $(u,v) \notin M$, implying that $M(T) = M(T')$ (because $T$ and $T'$ are identically distributed outside of the $(u,v)$ match). $T|_M$ and $T'|_M$ therefore have the same sets of teams, but it can still be that $p^T_{uv} \neq p^{T'}_{uv}$, so the tournaments $T|_M$ and $T'|_M$ can still differ due to this match (if both $u$ and $v$ are not eliminated).

The second step in our framework simply splits the recursive analysis into cases based on $M$, and the results of the matches $M(T)$. We define these cases clearly below, and then state our main framework. 

\begin{definition}[Base Case] We say that tournament $T$ is a \emph{base case} if $\alpha_S^r(T) = 0$. That is, we have hit a base case when it is not possible for $S$ to gain by manipulating tournament $T$ under rule $r$.
\end{definition}

One clear base case occurs if $S=\{u,v\}$, but $u$ is not even in $T$. Lemma~\ref{lem:anon} later identifies another base case for anonymous tournament rules (essentially: if relabeling $u$ and $v$ doesn't change $T$ except for the $(u,v)$ match, then $\{u,v\}$ cannot gain by manipulating an anonymous tournament rule).

\begin{definition}[Bad Terminal Event] $M$ is a \emph{bad terminal event} if $(u,v) \in M$. That is, a bad terminal event occurs when the $(u,v)$ match is played this round. We denote the occurence of this event by $B$.
\end{definition}

\begin{definition}[Good Terminal Event] $M, M(T)$ is a \emph{good terminal event} if $(u,v) \notin M$, and $M(T)$ is such that $T|_M$ is a base case.\footnote{Observe that because $(u,v) \notin M$, that if $T'$ is $S$-adjacent to $T$ then $M(T)$ and $M(T')$ are identically distributed.} That is, the $(u,v)$ match is not played now (so no gains from manipulation are possible), and no future gains are possible either. We denote the occurence of this event by $G$.
\end{definition}

\begin{definition}[Recursive Event] $M, M(T)$ is a \emph{recursive event} if $(u,v) \notin M$, but $M(T)$ is not such that $T|_M$ is a base case. That is, a recursive event occurs when the $(u,v)$ match is not played now, but may be played later (and manipulating it may be beneficial). We denote the occurence of this event by $R$.
\end{definition}

With these definitions in mind, we now state the theorem which drives our analysis. Observe that the bound claimed by Theorem~\ref{thm:recursive} could in principle be (very) loose, but our subsequent sections show that this framework suffices to nail down $\alpha_2^r(\mathcal{T}^\varepsilon)$ for RDM and RSEB, and also that these are the best possible guarantees for any Condorcet-consistent tournament rule.

\begin{theorem}\label{thm:recursive} Let $r^E$ be a recursive tournament rule,  $\varepsilon \in [0,1/2]$, and $S:=\{u,v\}$. Let also $b,g, c \geq 0$ (with $b+g> 0$) be such that for \emph{any} $T \in \mathcal{S}^\varepsilon$, and \emph{any} $T'$ which is $S$-adjacent to $T$:
\begin{itemize}
\item $\Pr_{M\leftarrow E(n)}[B] \leq b$.
\item $\Pr_{M \leftarrow E(n), M(T)}[G] \geq g$.
\item $\mathbb{E}_{M \leftarrow E(n)}[r_S^E(T'|_M) - r_S^E(T|_M) | B] \leq c$.
\end{itemize} 

Then: $\alpha_2^{r^E}(\mathcal{T}^\varepsilon) \leq \frac{bc}{b+g}$.
\end{theorem}

\begin{proof} We prove the theorem by induction, and focus on $\mathcal{S}^\varepsilon$ first (extending to $\mathcal{T}^\varepsilon$ only at the end using Proposition~\ref{prop:convex}). As a base case, observe that when there are at most two teams remaining, no gains from manipulation are possible and therefore $\alpha_2^{r^E}(\mathcal{S}^\varepsilon_{\leq 2}) \leq 0 \leq \frac{bc}{b+g}$ as desired. 

For the inductive hypothesis, assume that $\alpha_2^{r^E}(\mathcal{S}^\varepsilon_{\leq n-1}) \leq \frac{bc}{b+g}$ and consider now any tournament $T\in \mathcal{S}^\varepsilon$ on $n$ teams, and an $S$-adjacent $T'$. We have the following chain of equalities, which essentially just breaks down $r_S^E(T') - r_S^E(T)$ based on the three events $R,B,G$ (below, $I(X)$ denotes the indicator random variable for event $X$, which is $1$ when event $X$ occurs and $0$ otherwise):

\begin{align*}
r_S^E(T') - r_S^E(T) &=\mathbb{E}_{M\leftarrow E(n)}[r^E_S(T'|_M) - r^E_S(T|_M)]\\
&=\mathbb{E}_{M \leftarrow E(n),M(T)}[(r_S^E(T'|_M) - r_S^E(T|_M)) \cdot (I(R) + I(B) + I(G))]\\
&= \Pr_{M\leftarrow E(n), M(T)}[R]\cdot \mathbb{E}_{M \leftarrow E(n), M(T)}[r_S^E(T'|_M) - r_S^E(T|_M)|R]\\
&\quad +\Pr_{M\leftarrow E(n)}[B]\cdot \mathbb{E}_{M \leftarrow E(n)}[r_S^E(T'|_M) - r_S^E(T|_M)|B]\\
&\quad + \Pr_{M\leftarrow E(n), M(T)}[G]\cdot \mathbb{E}_{M \leftarrow E(n), M(T)}[r_S^E(T'|_M) - r_S^E(T|_M)|G]
\end{align*}
The first line is just restating Lemma~\ref{lem:recurse}. The second line simply observes that exactly one of the events $R, B, G$ occur, and also uses linearity of expectation to take an expectation also over $M(T)$. The third line just observes that $\mathbb{E}[Y \cdot I(X)] = \Pr[X] \cdot \mathbb{X}[Y | X]$ for any random variable $Y$ and event $X$, breaks the sum into three parts (again using linearity of expectation), and observes that the event $B$ is determined entirely by $M$ and is independent of $M(T)$. 

We now want to analyze the three terms separately. First, observe that by bullet three of the hypotheses: 

\begin{equation}\label{eq:one}
\mathbb{E}_{M \leftarrow E(n)}[r_S^E(T'|_M) - r_S^E(T|_M)|B] \leq c.
\end{equation}

Next, observe that in either a good terminal event or a recursive event, $(u,v) \notin M$. Because $T, T'$ are $S$-adjacent, this means that $M(T),M(T')$ are identically distributed (and can therefore be coupled so that $M(T) = M(T')$ with probability one), so the teams in $T|_M$ and $T'|_M$ are therefore the same. Finally, because $T$ is an \emph{independent} probabilistic tournament, and at least one team that participates in every match in $M$ is eliminated, $p^{T|_M}_{ij}= p^T_{ij}$ for all teams $i,j$ which are present in $T|_M$ (and also $p^{T'|_M}_{ij}= p^{T'}_{ij}$). This means that $T|_M \in \mathcal{S}^\varepsilon_{\leq n-1}$, and also that $T'|_M,T|_M$ are $S$-adjacent.

In a good terminal event, because $\alpha_2^{r^E}(T|_M) = 0$ by definition, we must therefore have:

\begin{equation}\label{eq:two}
\mathbb{E}_{M \leftarrow E(n), M(T)}[r_S^E(T'|_M) - r_S^E(T|_M)|G] \leq 0.
\end{equation}

In a recursive event, we instead have:

\begin{equation}\label{eq:three}
\mathbb{E}_{M \leftarrow E(n), M(T)}[r_S^E(T'|_M) - r_S^E(T|_M)|R]
\leq \alpha_2^{r^E}(\mathcal{S}_{\leq n-1}).
\end{equation}

Finally, we may now use our inductive hypothesis and Equations~\eqref{eq:one},~\eqref{eq:two},~\eqref{eq:three} to conclude:

\begin{align*}
r_S^E(T') - r_S^E(T) &= \Pr_{M\leftarrow E(n), M(T)}[R]\cdot \mathbb{E}_{M \leftarrow E(n), M(T)}[r_S^E(T'|_M) - r_S^E(T|_M)|R]\\
&\quad +\Pr_{M\leftarrow E(n)}[B]\cdot \mathbb{E}_{M \leftarrow E(n)}[r_S^E(T'|_M) - r_S^E(T|_M)|B]\\
&\quad + \Pr_{M\leftarrow E(n), M(T)}[G]\cdot \mathbb{E}_{M \leftarrow E(n), M(T)}[r_S^E(T'|_M) - r_S^E(T|_M)|G]\\
& \leq \Pr_{M\leftarrow E(n), M(T)}[R]\cdot \frac{bc}{b+g}+\Pr_{M\leftarrow E(n)}[B]\cdot c+ \Pr_{M\leftarrow E(n), M(T)}[G]\cdot 0\\
&= \frac{bc}{b+g}+\Pr_{M\leftarrow E(n)}[B]\cdot \left(c-\frac{bc}{b+g}\right)+ \Pr_{M\leftarrow E(n), M(T)}[G]\cdot \left(0 - \frac{bc}{b+g}\right)\\
& \leq (1-b-g)\cdot \frac{bc}{b+g}+b\cdot c+ g\cdot 0\\
&= \frac{bc}{b+g}
\end{align*}


The first inequality just combines the work of Equations~\eqref{eq:one},~\eqref{eq:two},~\eqref{eq:three}, and upper bounds the expected gains in all three cases using either the inductive hypothesis ($R$), direct hypothesis ($B$), or definition of good terminal event ($G$). The final inequality observes that $c \geq \frac{bc}{b+g}\geq 0$, so the bound is maximized when $B$ occurs as often as possible, while $G$ occurs as little as possible (consistent with the hypotheses).

We have now shown that for any $T \in \mathcal{S}^\varepsilon$, and any $T'$ which is $S$-adjacent, that $r_S^E(T') - r_S^E(T) \leq \frac{bc}{b+g}$. This establishes that $\alpha_2^{r^E}(\mathcal{S}^\varepsilon) \leq \frac{bc}{b+g}$. Proposition~\ref{prop:convex} extends this to $\mathcal{T}^\varepsilon$ as well.
\end{proof}
Theorem~\ref{thm:recursive} is our main framework for analysis. The remainder of this paper now computes the bounds required for the three bullets for the two tournament rules of interest as a function of $\varepsilon$, and substitutes to obtain tight bounds. Finally, it is worth briefly noting that the definition of $\alpha_2^r(\mathcal{T}^\varepsilon)$ semantically assumes that teams must decide how to manipulate the outcome of their match \emph{in advance}, before any matches are played. Still, any analysis that follows from Theorem~\ref{thm:recursive} applies \emph{even to manipulations which are decided upon as the match is played}. This is because the bound in bullet three of Theorem~\ref{thm:recursive} must hold over \emph{all} $S$-adjacent $T'$, not just one which was decided in advance. This claim is not central to our main results, and making it formal would be notationally cumbersome, so we provide only this brief discussion.

\section{Warmup: Rederiving Bounds for Deterministic Tournaments}\label{sec:deterministic}
In this section, we rederive the main result of~\cite{SchneiderSW17} (RSEB is $2$-SNM-$1/3$), and analyze a novel tournament rule (RDM is $2$-SNM-$1/3$) using a simple application of Theorem~\ref{thm:recursive}. This will require one further lemma about anonymous tournament rules to find additional base cases. The proof is in Appendix~\ref{app:warmup}.

\begin{lemma}\label{lem:anon} Let $r$ be any anonymous tournament rule, $S:=\{u,v\}$, and $T, T'$ be independent probabilistic tournaments which are $S$-adjacent \emph{and satisfy $p^T_{uw} = p^T_{vw}$ for all $w \notin \{u,v\}$}. Then $r_S(T) = r_S(T')$.
\end{lemma}

We now prove that RDM is $2$-SNM-$1/3$. Recall that $\alpha_2(\mathcal{S}^{1/2}) = 1/3$, so this is the best possible.

\begin{theorem}\label{thm:RDMeasy}RDM is $2$-SNM-$1/3$. Or in our language, $\alpha_2^{\text{RDM}}(\mathcal{T}^{1/2}) = 1/3$.
\end{theorem}
\begin{proof}
Recall that we need to lower bound the probability of a good terminal event, upper bound the probability of a bad terminal event, and upper bound the gains from manipulation in case of a bad terminal event. For the deterministic case, the latter bound is particularly simple, and we will just observe that clearly $\mathbb{E}_{M \leftarrow E(n)}[r_S^E(T'|_M) - r_S^E(T|_M) | B] \leq 1$. This is simply because the maximum probability that any coalition can win in any tournament is $1$, and the minimum is $0$. So we have established that RDM (and in fact any tournament rule) satisfies $c=1$.

To bound the probability of a bad terminal event, observe that a bad terminal event occurs only when $u$ plays $v$ in this round, which happens with probability exactly $1/\binom{n}{2}$. So RDM satisfies $b = 1/\binom{n}{2}$.

For the good terminal events, we claim that $\Pr[G] \geq 2/\binom{n}{2}$. First, observe that once we show this, we can plug into Theorem~\ref{thm:recursive} and conclude $\alpha_2^{\text{RDM}}(\mathcal{T}^{1/2}) \leq \frac{1\cdot 1/\binom{n}{2}}{1/\binom{n}{2}+2/\binom{n}{2}} = 1/3$. 

To see this bound, let $\ell_u$ denote the number of teams which beat $u$ but not $v$, and $\ell_v$ denote the number of teams which beat $v$ but not $u$. Without loss of generality let $\ell_u \geq \ell_v$. 

If $\ell_u +\ell_v= 0$, then Lemma~\ref{lem:anon} already establishes that $u$ and $v$ can gain nothing by manipulating. If $\ell_u +\ell_v \geq 2$, then whenever $u$ or $v$ play a team which beat them, we have a good terminal event (because $u$ or $v$ is already eliminated before having ever played the $(u,v)$ match, resulting in a base case). This happens with probability at least $2/\binom{n}{2}$, as desired.

If $\ell_u + \ell_v = 1$, then $\ell_u = 1, \ell_v = 0$. Let $w$ be the unique team which beats $u$ but not $v$. We claim that if $w$ plays \emph{either} $u$ or $v$ that we are in a good terminal event. Indeed, if $w$ plays $u$, then $u$ is eliminated having never played the $(u,v)$ match. If instead $w$ plays $v$, then $w$ is eliminated, but now there are no remaining teams which beat $u$ but not $v$ (or vice versa) and Lemma~\ref{lem:anon} asserts that there are no further gains from manipulation. The probability that $w$ plays $u$ or $v$ is $2/\binom{n}{2}$, as desired.

This handles all possible cases, and establishes that $g \geq 2/\binom{n}{2}$ in all cases. Plugging into Theorem~\ref{thm:recursive} as described above completes the proof.
\end{proof}

The analysis of RSEB is extremely similar to RDM, and requires only slightly more calculations to bound the probability of a good terminal event. This proof structure is fairly different than the original analysis in~\cite{SchneiderSW17}, and highlights the similarities to other recursive tournament rules. A complete proof of Theorem~\ref{thm:RSEBeasy} appears in Appendix~\ref{app:warmup}.

\begin{theorem}[\cite{SchneiderSW17}]\label{thm:RSEBeasy} RSEB is $2$-SNM-$1/3$. Or in our language, $\alpha_2^{\text{RSEB}}(\mathcal{T}^{1/2}) = 1/3$.
\end{theorem}

We wrap up our warmup by highlighting key points of the analysis which will be relevant for our main results. First, observe that our analysis of both rules succeeded by simply upper-bounding $c$ by $1$. Improving this as a function of $\varepsilon$ is the biggest technical difference between our warmup and main results. 

Second, observe that our analysis required Lemma~\ref{lem:anon} in case $\ell_u < 2$. In particular, we needed good terminal events \emph{even when both $u$ and $v$ were not eliminated} (and this need continues in main results).
\section{Optimal Bounds for Independent Probabilistic Tournaments}\label{sec:epsilon}
Before analyzing our two tournament rules, we extend the simple 3-team lower bound of~\cite{SchneiderSW17} for $\alpha_2(\mathcal{T}^{1/2})$ to $\alpha_2(\mathcal{T}^\varepsilon)$. The proof is in Appendix~\ref{app:epsilon}.

\begin{lemma}\label{lem:lb}
$\alpha_2(\mathcal{T}^\varepsilon) \geq \tfrac{1}{3}\varepsilon + \tfrac{2}{3}\varepsilon^2$.
\end{lemma}

\subsection{Gauntlets: Upper Bounding Gains from Bad Terminal Events}\label{sec:badterminal}
As previously noted, the main difference between our warmup and main results is bounding gain from bad terminal events. We provide a short, but key, structural insight about recursive tournament rules. Intuitively, a team $u$ wins under rule $r^E$ as long as they survive all elimination matches. Our key observation is that this defines a \emph{gauntlet} of teams such that $u$ wins if and only if they defeat every team in the gauntlet.

\begin{definition}[Gauntlet] For deterministic tournament $T$, recursive tournament rule $r$, and team $u$, let $T'$ be such that the outcome of the $(v,w)$ match is the same for all $v,w \neq u$, but $u$ is a Condorcet winner.

The \emph{gauntlet} for $u$ in tournament $T$ under recursive rule $r$, $G^r_u(T)$, is the set of teams that $u$ plays in elimination matches when $r$ is executed on $T'$. If $r$ is randomized, then $G^r_u(T)$ is a random variable. 

If $T$ is a probabilistic tournament, we extend the notation $G^r_u(T)$ to be the random variable which first samples $T$, then outputs $G^r_u(T)$ (again over randomness in $r$). 
\end{definition}


For intuition, consider RSEB. $u$ wins RSEB if and only if they win each of their $\lceil \log_2(n) \rceil$ matches, so their gauntlet is a list of $\lceil \log_2(n) \rceil$ teams, one per round. For RDM, however, the set of matches that $u$ plays is itself a random variable (depending on how many times $u$ is selected to play), but $u$ still must win all these matches in order to win. For RKotH, the size of the gauntlet is a random variable as well. Importantly, however, observe that for all three rules (and any elimination rule), as soon as $u$ loses a match, they are eliminated, so their gauntlet opponents can be set assuming that $u$ won all previous matches.

Importantly, observe that $u$'s gauntlet does \emph{not} depend on the outcome of any of its own matches (because $T'$ immediately causes $u$ to win all its matches anyway). This lets us make the following key observation, which is the main part of our analysis which requires $E$ to be a matching elimination rule.

\begin{lemma}\label{lem:anonrecurse} Let $r^E$ be an anonymous recursive tournament rule \emph{with a matching elimination rule}, $T$ be an independent probabilistic tournament on $n$ teams, $M\leftarrow E(n)$, and $u,v$ be two teams. Let also $(u,v) \in M$ and  $w$ denote the winner of the $(u,v)$ match. Then the random variable $G_w^{r^E}(T|_M)$ is \emph{independent of $w$}.
\end{lemma}
\begin{proof}
Observe first that because $M$ is a matching which contains $(u,v)$, that exactly one of $\{u,v\}$ (namely, $w$) is present in $T|_M$. Moreover, because $T$ is independent probabilistic, the remaining teams in $T|_M$ are independent of $w$. Because the definition of $w$'s gauntlet immediately considers a tournament $T'$ which replaces the outcome of all matches involving $w$ by having $w$ be a Condorcet winner, the tournament $T'$ is independent of $w$ \emph{except for whether $w$ is labeled as `$u$' or `$v$'}. But because $r^E(\cdot)$ is anonymous, its behavior on $T'$ is independent of $w$'s label. This completes the proof.
\end{proof}

Note that Lemma~\ref{lem:anonrecurse} fails when $M$ is not a matching. This is because: (a) perhaps both $u$ and $v$ are eliminated, and $w$ is undefined, but also (b) even conditioned on $w$ being defined, the set of teams in $T|_M$ can depend on $w$. This is the main technical challenge in extending beyond matching elimination rules.\footnote{For example, in RKotH: conditioned on $w:=u$, we know either that $u$ was a Condorcet winner, or that $v$ was selected to play everyone, so all teams which lose to $v$ are eliminated. If instead $w:=u$, we know either that $v$ was a Condorcet winner, or that $u$ was selected to play everyone. The teams which lose to $u$ vs. $v$ could be different, so the Lemma fails to hold.}

\begin{corollary}\label{cor:anonrecurse} Let $r^E$ be an anonymous tournament rule with a matching elimination rule, $T\in \mathcal{T}^\varepsilon$, $S = \{u,v\}$ be any two teams, and $T'$ be any tournament which is $S$-adjacent to $T$. Then:

$$\mathbb{E}_{M \leftarrow E(n)}[r_S^E(T'|_M) - r_S^E(T|_M) | B] \leq 2\varepsilon(\tfrac{1}{2}+\varepsilon).$$

\end{corollary}
\begin{proof}
Because we are in a bad terminal event, this means that $(u,v) \in M$. If we let $w$ denote the winner of the $(u,v)$ match in $T$, and $w'$ denote the winner in $T'$, then we know that $G^{r^E}_{w}(T|_M)$ and $G^{r^E}_{w'}(T'|_M)$ are identically distributed by Lemma~\ref{lem:anonrecurse}. This means that we have:
\begin{align*}
r^E_S(T|_M)&=\mathbb{E}_{G_{w}^{r^E}(T|_M)}\left[\prod_{x \in G_{w}^{r^E}(T|_M)} p^{T}_{wx}\right]\\
r^E_S(T'|_M)&=\mathbb{E}_{G_{w'}^{r^E}(T'|_M)}\left[\prod_{x \in G_{w'}^{r^E}(T'|_M)} p^{T'}_{w'x}\right] =\mathbb{E}_{G_{w}^{r^E}(T|_M)}\left[\prod_{x \in G_{w}^{r^E}(T|_M)} p^{T}_{w'x}\right]\\
\Rightarrow r^E_S(T'|_M)-r^E_S(T|_M) &= \mathbb{E}_{G_{w}^{r^E}(T|_M)}\left[\prod_{x \in G_{w}^{r^E}(T|_M)} p^{T}_{w'x} -\prod_{x \in G_{w}^{r^E}(T|_M)} p^{T}_{wx}\right]\\
&\leq \mathbb{E}_{G_{w}^{r^E}(T|_M)}\left[\prod_{x \in G_{w}^{r^E}(T|_M)} (\tfrac{1}{2}+\varepsilon) -\prod_{x \in G_{w}^{r^E}(T|_M)} (\tfrac{1}{2}-\varepsilon)\right]\\
&= \mathbb{E}_{G_{w}^{r^E}(T|_M)}\left[(\tfrac{1}{2}+\varepsilon)^{|G_{w}^{r^E}(T|_M)|}-(\tfrac{1}{2}-\varepsilon)^{|G_{w}^{r^E}(T|_M)|}\right]\\
&\leq 2\varepsilon
\end{align*}

The first two lines follow by definition of the gauntlet, and Lemma~\ref{lem:anonrecurse}. The third line is basic algebra. The fourth line follows as $T \in \mathcal{T}^{\varepsilon}$. The fifth line is again basic algebra, and the final line invokes (a special case of) Lemma~\ref{lem:calcs}, which is stated below and proved in Appendix~\ref{app:epsilon}. The above calculations hold for any $T|_M, T'|_M$. 

\begin{lemma}\label{lem:calcs}
For all $n \in \mathbb{N}_{\geq 0}$, $i, j\in [0,n]$ and $\varepsilon \in [0,1/2]$: $(\tfrac{1}{2}+\varepsilon)^i \cdot (\tfrac{1}{2}-\varepsilon)^{n-i} - (\tfrac{1}{2}+\varepsilon)^j \cdot (\tfrac{1}{2}-\varepsilon)^{n-j} \leq 2\varepsilon$.
\end{lemma}

To see where the additional factor of $(\tfrac{1}{2}+\varepsilon)$ comes from, recall that $p^T_{uv} \in [1/2-\varepsilon,1/2+\varepsilon]$. Therefore, $p^{T'}_{uv} - p^T_{uv} \leq \tfrac{1}{2}+\varepsilon$. Consider now coupling the tournaments $T$ and $T'$ so the the result of every match except for the one between $u$ and $v$ is identical (this is possible because $T$ and $T'$ are independent probabilistic). Under this coupling, $T$ and $T'$ are identical with probability at least $\tfrac{1}{2}-\varepsilon$ (and gains from manipulation are clearly only possible when $T \neq T'$, which occurs with probability at most $\tfrac{1}{2}+\varepsilon$, and independently of $M$). Therefore, $T\neq T'$ with probability at most $(\tfrac{1}{2}+\varepsilon)$ and conditioned on this, $r^E_S(T'|_M)-r^E_S(T|_M) \leq 2\varepsilon$.
\end{proof}

\subsection{RDM and RSEB}
Corollary~\ref{cor:anonrecurse} is the main technical lemma to extend our analysis for $\mathcal{T}^{1/2}$ to $\mathcal{T}^{\varepsilon}$. We begin with RDM.

\begin{theorem}\label{thm:RDMmain} $\alpha_2^{\text{RDM}}(\mathcal{T}^{\varepsilon}) = \alpha_2(\mathcal{T}^\varepsilon) =\tfrac{1}{3} \varepsilon + \tfrac{2}{3} \varepsilon^2 $.
\end{theorem}
\begin{proof}
Let $T \in \mathcal{S}^\varepsilon$. Recall that we need to lower bound the probability of a good terminal event, upper bound the probability of a bad terminal event, and upper bound the gains from manipulation in case of a bad terminal event. We have already upper bounded the gains from a bad terminal event using Corollary~\ref{cor:anonrecurse} and can therefore take $c = 2\varepsilon(\tfrac{1}{2}+\varepsilon)$.

To bound the probability of a bad terminal event, observe that a bad terminal event occurs only when $u$ plays $v$ in this round, which happens with probability exactly $1/\binom{n}{2}$. So RDM satisfies $b = 1/\binom{n}{2}$.

For the good terminal events, we claim that $\Pr[G] \geq 2/\binom{n}{2}$. First, observe that once we show this, we can plug into Theorem~\ref{thm:recursive} and conclude $\alpha_2^{\text{RDM}}(\mathcal{T}^{\varepsilon}) \leq \frac{2\varepsilon(\tfrac{1}{2}+\varepsilon)\cdot 1/\binom{n}{2}}{1/\binom{n}{2}+2/\binom{n}{2}} = \tfrac{1}{3} \varepsilon + \tfrac{2}{3} \varepsilon^2 $.

To see this bound, let $\ell_u$ denote the number of teams which beat $u$ with probability $(1/2+\varepsilon)$ but $v$ with probability $(1/2-\varepsilon)$, and $\ell_v$ denote the number of teams which beat $v$ with probabiltiy $(1/2+\varepsilon)$ but $u$ with probability $(1/2-\varepsilon)$. Without loss of generality let $\ell_u \geq \ell_v$. 

If $\ell_u = 0$, then Lemma~\ref{lem:anon} already establishes that $u$ and $v$ can gain nothing by manipulating. If $\ell_u +\ell_v \geq 2$, then whenever $u$ or $v$ plays a team (not in $\{u,v\}$) which beats them, we have a good terminal event (because $u$ or $v$ is already eliminated before having ever played the $(u,v)$ match, so manipulating the match has no impact). Observe that any team which beats $u$ (respectively, $v$) with probability $(1/2+\varepsilon)$ beats $v$ (respectively, $u$) with probability at least $(1/2-\varepsilon)$. Therefore, we have a good terminal event with probability at least $2(1/2+\varepsilon)/\binom{n}{2} + 2(1/2-\varepsilon)/\binom{n}{2}= 2/\binom{n}{2}$, as desired.\footnote{To be extra clear, the two matches $(w,u)$, $(w,v)$ together contribute probability $1/\binom{n}{2}$ to the probability of a good terminal event, as long as $w$ beats either $u$ or $v$ with probability $(1/2+\varepsilon)$. Because there are two such teams, we get this twice.}

If $\ell_u = 1, \ell_v = 0$, then let $w$ be the unique team which beats $u$ with probability $(1/2+\varepsilon)$ but $v$ with probability $(1/2-\varepsilon)$. We claim that if $w$ plays \emph{either} $u$ or $v$ that we have a good terminal event. Indeed, if $w$ wins, then either $u$ or $v$ are eliminated having never played the $(u,v)$ match. If instead $w$ loses this match, then $w$ is eliminated, but now $p^T_{ux} = p^T_{vx}$ for all remaining teams $x$, and Lemma~\ref{lem:anon} asserts that there are no further gains from manipulation. The probability that $w$ plays $u$ or $v$ is $2/\binom{n}{2}$, as desired.

This handles all possible cases, and establishes that $g \geq 2/\binom{n}{2}$ in all cases. Plugging into Theorem~\ref{thm:recursive} as described above completes the proof.
\end{proof}

The analysis for RSEB is again similar to RDM, but some calculations are more involved. 

\begin{theorem}\label{thm:RSEBmain} $\alpha_2^{\text{RSEB}}(\mathcal{T}^{\varepsilon}) = \alpha_2(\mathcal{T}^\varepsilon) =\tfrac{1}{3} \varepsilon + \tfrac{2}{3} \varepsilon^2 $.
\end{theorem}
\begin{proof}
Let $T \in \mathcal{S}^\varepsilon$. We have already upper bounded the gains from a bad terminal event using Corollary~\ref{cor:anonrecurse} and can therefore take $c = 2\varepsilon(\tfrac{1}{2}+\varepsilon)$.

To bound the probability of a bad terminal event, observe that a bad terminal event occurs only when $u$ plays $v$ this round, which happens with probability exactly $1/(n'-1)$.\footnote{Recall in RSEB that $n'$ denotes the total number of teams plus dummy teams, and is $2^{\lceil \log_2(n)\rceil}$}. So RSEB satisfies $b = 1/(n'-1)$.

For the good terminal events, we claim that $\Pr[G] \geq 2/(n'-1)$. First, observe that once we show this, we can plug into Theorem~\ref{thm:recursive} and conclude $\alpha_2^{\text{RSEB}}(\mathcal{T}^{\varepsilon}) \leq \frac{2\varepsilon(\tfrac{1}{2}+\varepsilon)\cdot 1/(n'-1)}{1/(n'-1)+2/(n'-1)} = \tfrac{1}{3} \varepsilon + \tfrac{2}{3} \varepsilon^2 $.

To see this bound (which requires more calculations than previous proofs), let $\ell_u$ denote the number of teams which beat $u$ with probability $(1/2+\varepsilon)$ but $v$ with probability $(1/2-\varepsilon)$, and $\ell_v$ denote the number of teams which beat $v$ with probabiltiy $(1/2+\varepsilon)$ but $u$ with probability $(1/2-\varepsilon)$. Without loss of generality let $\ell_u \geq \ell_v$.\\ 

\noindent\textbf{Case One: $\ell_u +\ell_v \leq 1$.} If $\ell_u = 0$, then Lemma~\ref{lem:anon} already establishes that $u$ and $v$ can gain nothing by manipulating. If $\ell_u = 1, \ell_v = 0$, then let $w$ be the unique team which beats $u$ with probability $(1/2+\varepsilon)$ but $v$ with probability $(1/2-\varepsilon)$. We claim that if $w$ plays \emph{either} $u$ or $v$ that we are in a good terminal event. Indeed, if $w$ wins this match, then either $u$ or $v$ are eliminated having never played the $(u,v)$ match. If instead $w$ loses this match, then $w$ is eliminated, but now $p^T_{ux} = p^T_{vx}$ for all remaining teams $x$, and Lemma~\ref{lem:anon} asserts that there are no further gains from manipulation. The probability that $w$ plays $u$ or $v$ is $2/(n'-1)$, as desired. \\


\noindent\textbf{Case Two: $\ell_u + \ell_v = 2$.} Next, consider the case where $\ell_u + \ell_v = 2$, and call the two relevant teams $x, w$. Observe first that if $u$ plays $x$ and $v$ plays $w$, or if $u$ plays $w$ and $v$ plays $x$, then we are surely in a good terminal event. This is because either (a) $u$ or $v$ is eliminated without having played the $(u,v)$ match, or (b) \emph{both} $x$ and $w$ are eliminated (allowing us to invoke Lemma~\ref{lem:anon}). This occurs with probability $\frac{2}{(n'-1)(n'-3)}$. There are also the cases where exactly one of $\{w,x\}$ plays a team in $\{u,v\}$. For any given pair $(a,b)$, with $a \in \{w,x\}$ and $b \in \{u,v\}$, this case occurs with probability $\frac{(n'-4)}{(n'-1)(n'-3)}$.\footnote{Because the probability that $a$ plays $b$ is $1/(n'-1)$, and the probability that the other two teams do not play, conditioned on this, is $(n'-4)/(n'-3)$.} Two of these cases contribute at least a $(1/2+\varepsilon)$ probability of eliminating the team in $\{u,v\}$, and the other two contribute at least a $(1/2-\varepsilon)$ probability. So in total, all four cases contribute at least $\frac{2(n'-4)}{(n'-1)(n'-3)}$, and together we get that a good terminal event occurs with probability at least:
\begin{align*}
\frac{2}{(n'-1)(n'-3)}+\frac{2(n'-4)}{(n'-1)(n'-3)}&= \frac{2(n'-3)}{(n'-1)(n'-3)} = \frac{2}{n'-1}
\end{align*}

\noindent\textbf{Case Three: $\ell_u + \ell_v \geq 3$.} Next, consider the case where $\ell_u + \ell_v \geq 3$ (observe that this implies $n \geq 5$). We will show that either $u$ or $v$ are eliminated with probability at least $2/(n'-1)$. Indeed, let $L_u$ denote the set of teams which beat $u$ with probability $(1/2+\varepsilon)$ (but not $v$), and $L_v$ denote the set of teams which beat $v$ with probability $(1/2+\varepsilon)$ (but not $u$). Then consider the case where $u$ plays a team in $L_u$, or $v$ plays a team in $L_v$. Conditioned on this, \emph{both} $u$ and $v$ survive with probability at most $(1/2+\varepsilon)(1/2-\varepsilon) \leq 1/4$, so one of $\{u,v\}$ is eliminated with probability at least $3/4$. So one sufficient condition would be to establish that $u$ plays a team in $L_u$ or $v$ plays a team in $L_v$ with probability at least $\frac{8/3}{n'-1}$ (because conditioned on this, one of $\{u,v\}$ is eliminated with probability $3/4$, for a total probability of at least $\frac{2}{n'-1}$ that one of $\{u,v\}$ is eliminated). In the subcase that $\ell_u \geq 3$, then the probability is in fact at least $3/(n'-1)$, as desired. Clearly, this probability is monotone in $\ell_u, \ell_v$, so this leaves the only remaining subcase as $\ell_u = 2, \ell_v = 1$.\\

\noindent\textbf{Subcase Three-A: $\ell_u = 2, \ell_v = 1$, $n \geq 9$.} For the subcase of $\ell_u = 2, \ell_v=1$, the probability that $u$ plays a team in $L_u$ or $v$ plays a team in $L_v$ is $\frac{2}{n'-1}+\frac{(n'-5)}{(n'-1)(n'-3)}$.\footnote{To see this, observe that the probability that $u$ plays a team in $L_u$ is $2/(n'-1)$. The probability that $u$ plays a team not in $L_u \cup L_v \cup \{v\}$ is $(n'-5)/(n'-1)$. Conditioned on this, the probability that $v$ plays the team in $L_v$ is $1/(n'-3)$ (and this is the only way that $v$ can possibly play the team in $L_v$ without $u$ playing $L_u$).} As $n'\geq 16$ (because $n \geq 9$ and $n'$ is a power of $2$), we have that $(n'-5)/(n'-3) \geq 11/13$, meaning that $\frac{2}{n'-1}+\frac{(n'-5)}{(n'-1)(n'-3)} \geq \frac{37/13}{n'-1} > \frac{8/3}{n'-1}$, which resolves this case by the work above.\\

\noindent\textbf{Subcase Three-B: $\ell_u =2,\ell_v = 1$, $\varepsilon \geq \frac{1}{2\sqrt{13}}, n \in [5,8]$.} When $n \in [5,8]$, we have $n' = 8$. This means that $(n'-5)/(n'-3) = 3/5$, and therefore $\frac{2}{n'-1}+\frac{(n'-5)}{(n'-1)(n'-3)} = \frac{13/5}{n'-1}$. Unfortunately, this is $< \frac{8/3}{n'-1}$, meaning that this case doesn't immediately resolve by the method above. In this range, we do an explicit case analysis. First, observe that $u$ and $v$ are \emph{least} likely to be eliminated when there are more dummy teams (because $u$ and $v$ beat dummy teams with probability one, but real teams with probably at most $1/2+\varepsilon$). So the worst case to consider is when $n=5$: two real teams beat only $u$, one beats only $v$, and there are three dummy teams. We just have to compute several cases:
\begin{itemize}
\item Perhaps $u$ plays one of the $\ell_u$ teams, and $v$ plays one of the $\ell_v$ teams. This occurs with probability $\frac{2}{7} \cdot \frac{1}{5}$, and eliminates $u$ or $v$ with probability $1-(1/2-\varepsilon)^2 = 3/4 +\varepsilon - \varepsilon^2$.
\item Perhaps $u$ plays one of the $\ell_u$ teams, and $v$ plays the other. This occurs with probability $\frac{2}{7} \cdot \frac{1}{5}$, and eliminates $u$ or $v$ with probability $1-(1/2+\varepsilon)(1/2-\varepsilon) = 3/4+\varepsilon^2$. 
\item Perhaps $u$ plays one of the $\ell_u$ teams, and $v$ plays a dummy team. This occurs with probability $\frac{2}{7} \cdot \frac{3}{5}$, and eliminutes $u$ with probability $1/2+\varepsilon$.
\item Perhaps $u$ plays one of the $\ell_v$ teams, and $v$ plays one of the $\ell_u$ teams. This occurs with probability $\frac{1}{7} \cdot \frac{2}{5}$, and eliminates $u$ or $v$ with probability $1-(1/2+\varepsilon)^2 = 3/4-\varepsilon -\varepsilon^2$.
\item Perhaps $u$ plays one of the $\ell_v$ teams, and $v$ plays a dummy team. This occurs with probability $\frac{1}{7} \cdot \frac{3}{5}$, and eliminates $u$ with probability $1/2-\varepsilon$.
\item Perhaps $u$ plays a dummy team, and $v$ plays one of the $\ell_v$ teams. This occurs with probability $\frac{3}{7} \cdot \frac{1}{5}$, and eliminates $v$ with probability $1/2+\varepsilon$.
\item Perhaps $u$ plays a dummy team, and $v$ plays one of the $\ell_u$ teams. This occurs with probability $\frac{3}{7} \cdot \frac{2}{5}$, and eliminates $v$ with probability $1/2-\varepsilon$. 
\item Perhaps both $u$ and $v$ play dummy teams. This happens with probability $\frac{3}{7} \cdot \frac{2}{5}$, but eliminates neither team.
\end{itemize}

So either $u$ or $v$ is eliminated (without playing each other) with probability:
\begin{align*}
\frac{2}{35} \cdot (3/4+\varepsilon-\varepsilon^2) &+ \frac{2}{35} \cdot (3/4+\varepsilon^2) + \frac{6}{35} \cdot (1/2+\varepsilon) + \frac{2}{35} \cdot (3/4-\varepsilon-\varepsilon^2)\\
 &+ \frac{3}{35}\cdot (1/2-\varepsilon) + \frac{3}{35}\cdot (1/2+\varepsilon) + \frac{6}{35}\cdot (1/2-\varepsilon)\\
&=\frac{27}{70} -\frac{2}{35} \varepsilon^2\\
&\geq \frac{26}{70} \\
&> 2/7.
\end{align*}

Note that when $n \in [5,8]$, the probability of a bad event is exactly $1/7$, and the above work establishes that in case Three-B, the probability of a good terminal event is at least twice that of a bad terminal event. \\

The arguments above handle all possible cases, and establishes that $g \geq 2/(n'-1)$ in all cases. Plugging into Theorem~\ref{thm:recursive} as described above completes the proof.
\end{proof}

\section{Conclusion and Open Problems}\label{sec:conclusion}
We take a beyond worst-case view on manipulability of tournament rules, and nail down optimal guarantees as a function of the uncertainty of match outcomes. Specifically, our main result shows that $\alpha_2(\mathcal{T}^\varepsilon)=\varepsilon/3 + 2\varepsilon^2/3$, and this is achieved \emph{for all $\varepsilon$} by Randomized Death Match and Randomized Single Elimination Bracket. Our main technical contribution is a framework to analyze recursive tournament rules.

There are two natural directions for future work. The first concerns further work in the probabilistic setting: does $\alpha_2^{\text{RKotH}}(\mathcal{T}^\varepsilon) = \varepsilon/3+2\varepsilon^2/3$? The main technical barrier is replacing Lemma~\ref{lem:anonrecurse}, which only holds for matching elimination rules.\footnote{There are other barriers to using our precise definitions, but these barriers seem semantic rather than substantial.} In addition, it is interesting to analyze probabilistic tournaments which are not independent. Here there are technical barriers to overcome (many of our steps do require independence), but also conceptual ones: if tournament match outcomes are correlated, should we consider manipulations which are correlated with external outcomes as well? If so, is there a natural way to consider manipulations which are ``not more correlated than the original tournament itself''?

A second direction concerns applications of our recursive framework towards other problems in approximately strategy-proof tournament design. For example, it is still an open question following~\cite{SchvartzmanWZZ20} what is $\alpha_k(\mathcal{T}^{1/2})$ for \emph{any} $k > 2$. Our recursive framework may prove useful for analyzing this, or at least determining the best achievable guarantees for recursive rules.

\bibliographystyle{alpha}
\bibliography{MasterBib}

\newcommand{\etalchar}[1]{$^{#1}$}
\begin{thebibliography}{BCE{\etalchar{+}}16}

\bibitem[AK10]{AltmanK10}
Alon Altman and Robert Kleinberg.
\newblock Nonmanipulable randomized tournament selections.
\newblock 2010.

\bibitem[APT09]{AltmanPT09}
Alon Altman, Ariel~D. Procaccia, and Moshe Tennenholtz.
\newblock Nonmanipulable selections from a tournament.
\newblock In {\em Proceedings of the 21st International Joint Conference on
  Artifical Intelligence}, IJCAI'09, pages 27--32, San Francisco, CA, USA,
  2009. Morgan Kaufmann Publishers Inc.

\bibitem[Ban85]{Banks85}
J.~S. Banks.
\newblock Sophisticated voting outcomes and agenda control.
\newblock {\em Social Choice and Welfare}, 1(4):295--306, Dec 1985.

\bibitem[BCE{\etalchar{+}}16]{BrandtCELP16}
Felix Brandt, Vincent Conitzer, Ulle Endriss, J{\'e}r{\^o}me Lang, and Ariel~D
  Procaccia.
\newblock {\em Handbook of Computational Social Choice}.
\newblock Cambridge University Press, 2016.

\bibitem[Csa17]{Csato17}
Laszlo Csato.
\newblock European qualifiers 2018 {FIFA} world cup qualification can be
  manipulated.
\newblock 09 2017.

\bibitem[Fis77]{Fishburn77}
Peter~C. Fishburn.
\newblock Condorcet social choice functions.
\newblock {\em SIAM Journal on Applied Mathematics}, 33(3):469--489, 1977.

\bibitem[FR92]{FisherR92}
David~C. Fisher and Jennifer Ryan.
\newblock Optimal strategies for a generalized “scissors, paper, and stone”
  game.
\newblock {\em The American Mathematical Monthly}, 99(10):935--942, 1992.

\bibitem[KSW16]{KimSW16}
Michael~P. Kim, Warut Suksompong, and Virginia~Vassilevska Williams.
\newblock Who can win a single-elimination tournament?
\newblock In {\em Proceedings of the Thirtieth {AAAI} Conference on Artificial
  Intelligence, February 12-17, 2016, Phoenix, Arizona, {USA.}}, pages
  516--522, 2016.

\bibitem[KW15]{KimW15}
Michael~P. Kim and Virginia~Vassilevska Williams.
\newblock Fixing tournaments for kings, chokers, and more.
\newblock In {\em Proceedings of the Twenty-Fourth International Joint
  Conference on Artificial Intelligence, {IJCAI} 2015, Buenos Aires, Argentina,
  July 25-31, 2015}, pages 561--567, 2015.

\bibitem[Las97]{Laslier97}
J.F. Laslier.
\newblock {\em Tournament Solutions and Majority Voting}.
\newblock Studies in Economic Theory (Berlin, Germany), 7. Springer, 1997.

\bibitem[LLB93]{LaffondLB93}
G.~Laffond, J.F. Laslier, and M.~Le Breton.
\newblock The bipartisan set of a tournament game.
\newblock {\em Games and Economic Behavior}, 5(1):182 -- 201, 1993.

\bibitem[Mil80]{Miller80}
Nicholas~R. Miller.
\newblock A new solution set for tournaments and majority voting: Further
  graph-theoretical approaches to the theory of voting.
\newblock {\em American Journal of Political Science}, 24(1):68--96, 1980.

\bibitem[Pau14]{Pauly14}
Marc Pauly.
\newblock Can strategizing in round-robin subtournaments be avoided?
\newblock {\em Social Choice and Welfare}, 43(1):29--46, 2014.

\bibitem[SSW17]{SchneiderSW17}
Jon Schneider, Ariel Schvartzman, and S.~Matthew Weinberg.
\newblock Condorcet-consistent and approximately strategyproof tournament
  rules.
\newblock In {\em 8th Innovations in Theoretical Computer Science Conference,
  {ITCS} 2017, January 9-11, 2017, Berkeley, CA, {USA}}, pages 35:1--35:20,
  2017.

\bibitem[SW84]{ShepsleW84}
Kenneth~A. Shepsle and Barry~R. Weingast.
\newblock Uncovered sets and sophisticated voting outcomes with implications
  for agenda institutions.
\newblock {\em American Journal of Political Science}, 28(1):49--74, 1984.

\bibitem[SW11]{StantonW11}
Isabelle Stanton and Virginia~Vassilevska Williams.
\newblock Rigging tournament brackets for weaker players.
\newblock In {\em {IJCAI} 2011, Proceedings of the 22nd International Joint
  Conference on Artificial Intelligence, Barcelona, Catalonia, Spain, July
  16-22, 2011}, pages 357--364, 2011.

\bibitem[SWZZ20]{SchvartzmanWZZ20}
Ariel Schvartzman, S.~Matthew Weinberg, Eitan Zlatin, and Albert Zuo.
\newblock Approximately strategyproof tournament rules: On large manipulating
  sets and cover-consistence.
\newblock In Thomas Vidick, editor, {\em 11th Innovations in Theoretical
  Computer Science Conference, {ITCS} 2020, January 12-14, 2020, Seattle,
  Washington, {USA}}, volume 151 of {\em LIPIcs}, pages 3:1--3:25. Schloss
  Dagstuhl - Leibniz-Zentrum f{\"{u}}r Informatik, 2020.

\bibitem[Wil10]{Williams10}
Virginia~Vassilevska Williams.
\newblock Fixing a tournament.
\newblock In {\em Proceedings of the Twenty-Fourth {AAAI} Conference on
  Artificial Intelligence, {AAAI} 2010, Atlanta, Georgia, USA, July 11-15,
  2010}, 2010.

\end{thebibliography}

\appendix
\section{Omitted Proofs}
\subsection{Omitted Proofs from Section~\ref{sec:prelim}}\label{app:prelim}
\begin{proof}[Proof of Proposition~\ref{prop:convex}]
The second statement of the lemma clearly follows from the first. It is trivial to see that $\alpha^r_k(\mathcal{T}^\varepsilon) \geq \alpha^r_k(\mathcal{S}^\varepsilon)$, as $\mathcal{S}^\varepsilon\subseteq \mathcal{T}^\varepsilon$. To get intuition for why $\alpha_k^r(\mathcal{T}) \leq \alpha_k^r(\mathcal{S}^\varepsilon)$, observe that every probabilistic tournament in $\mathcal{T}^\varepsilon$ can be written as a convex combination of probabilistic tournaments in $\mathcal{S}^\varepsilon$. Indeed, this intuition drives the proof.

To see this formally, let $T$ be an arbitrary (independent probabilistic) tournament in $\mathcal{T}^\varepsilon$, and let $T'$ be any $S$-adjacent (independent probabilistic) tournament for some $|S| \leq k$. Then $p_{ij}^T \in [1/2-\varepsilon,1/2+\varepsilon]$ for all $i,j$, by definition. Consider the following procedure to jointly sample tournaments from $T, T'$. The procedure first constructs another pair of probabilistic tournaments $U, U'$, and then samples from these.
\begin{enumerate}
\item For all $(i,j)$, set $q_{ij}:=\frac{p^T_{ij}-(1/2-\varepsilon)}{2\varepsilon}$. Observe that $q_{ij}\cdot (1/2+\varepsilon) + (1-q_{ij})\cdot (1/2-\varepsilon) = p_{ij}^T$, and that $q_{ij} \in [0,1]$ as $p_{ij}^T \in [1/2-\varepsilon,1/2+\varepsilon]$.
\item For each $(i,j)$, independently, set $p_{ij}$ equal to $(1/2+\varepsilon)$ with probability $q_{ij}$ and equal to $(1/2-\varepsilon)$ with probability $(1-q_{ij})$. 
\item For all $(i,j)$, set $p^{U}_{ij}:=p_{ij}$.
\item For all $(i,j)$ such that $\{i,j\} \not \subseteq S$, set $p^{U'}:=p^U$. For all $(i,j)$ such that $\{i,j\} \subseteq S$, set $p^{U'}_{ij}:=p^{T'}_{ij}$.
\item Draw a pair of tournaments according to $U, U'$ (independently, say).
\end{enumerate}

Observe that this procedure correctly draws two tournament according to $T, T'$. Indeed, it is easy to see for each output tournament that the outcome of each match is independent, simply because they are independent in $U$ (resp. $U'$), and because the random variables $p^U_{ij}$ (resp. $p^{U'}_{ij}$) are independent. Moreover, we claim that the probability that $i$ beats $j$ is exactly $p_{ij}^T$. Indeed, the probability that $i$ beats $j$ in $U$ is: $\mathbb{E}[p_{ij}^{U}] = q_{ij}\cdot (1/2+\varepsilon) + (1-q_{ij})\cdot (1/2-\varepsilon) = p_{ij}^T$. In $U'$, when $\{i,j\} \subseteq S$, the probability that $i$ beats $j$ in $U'$ is clearly $p_{ij}^{T'}$ by definition. When $\{i,j\} \not \subseteq S$, the probability that $i$ beats $j$ is also $p_{ij}^T$, which is equal to $p_{ij}^{T'}$ as $T, T'$ are $S$-adjacent.

Additionally, observe that: (a) $U \in \mathcal{S}^\varepsilon$ and (b) $U$ and $U'$ are $S$-adjacent. Now, consider any tournament rule $r(\cdot)$:

\begin{align*}
\text{For all $i$: } r_i(T)&=\mathbb{E}[r_i(U)]\\
\text{For all $i$: } r_i(T') &= \mathbb{E}[r_i(U')]\\
\Rightarrow r_S(T') - r_S(T) &= \mathbb{E}[ r_S(U')-r_S(U)]\\
&\leq \alpha_k^r(\mathcal{S}^\varepsilon)
\end{align*}

Indeed, the first two lines are simply linearity of expectation, once the previous work confirms that going through $U, U'$ is a valid way to draw tournaments from $T, T'$. The third line then also follows by linearity of expectation. The final line follows as $U \in \mathcal{S}^\varepsilon$, and $U,U'$ are $S$-adjacent. This completes the proof, as we have now shown that $\alpha_k^r(\mathcal{T}^\varepsilon) \leq \alpha_k^r(\mathcal{S}^\varepsilon)$.
\end{proof}

\subsection{Omitted Proofs from Section~\ref{sec:recursive}}\label{app:recursive}

\begin{proof}[Proof of Lemma~\ref{lem:recurse}]
The proof follows immediately from linearity of expectation. For all $i$, the probability that $i$ wins in tournament $T$ under rule $r^E(\cdot)$ is simply the expected probability that $i$ wins in the tournament $T|_M$. Summing over all $i \in S$, and repeating this for $T'$ yields the lemma.
\end{proof}

\subsection{Omitted Proofs from Section~\ref{sec:deterministic}}\label{app:warmup}

\begin{proof}[Proof of Lemma~\ref{lem:anon}]
Let $\sigma(\cdot)$ denote the permutation which swaps $u$ and $v$. Consider any two deterministic tournaments $U, \sigma(U)$. Then because $r(\cdot)$ is anonymous, we have:
\begin{align*}
r_u(U) + r_v(U) &= r_{\sigma(u)}(\sigma(U)) + r_{\sigma(v)}(\sigma(U))\\
&=r_{v}(\sigma(U)) + r_{u}(\sigma(U))\\
\Rightarrow r_S(U) &= r_S(\sigma(U))
\end{align*}

Indeed, the first line simply applies anonymity, and the second line simply applies $\sigma$. Now let's return to $T, T'$ (which are independent probabilistic tournaments, rather than deterministic). Consider the following process to draw $T$ and $T'$ jointly:
\begin{enumerate}
\item To emphasize that $p^T_{uw} = p^T_{vw}$, for all $w \notin \{u,v\}$, denote by $p^T_w:=p^T_{uw}$.
\item Without loss of generality, let $p^T_{uv} \leq p^{T'}_{uv}$.
\item Draw the outcome of all matches involving two teams both $\notin \{u,v\}$. Set the outcome of these matches the same for $T$ and $T'$.
\item For all $w \notin \{u,v\}$, draw $q_{w1}$ and $q_{w2}$ iid and uniformly from $[0,1]$. Draw $q_{uv}$ independently and uniformly from $[0,1]$.
\item For all $w \notin \{u,v\}$, have $u$ beat $w$ if and only if $q_{w1} < p^T_{w}$. Have $v$ beat $w$ if and only if $q_{w2} < p^T_{w}$. Have $u$ beat $v$ if and only if $q_{uv} < p^T_{uv}$.
\item If $q_{uv} \notin [p^T_{uv},p^{T'}_{uv}]$, set $T':=T$.
\item If $q_{uv} \in [p^T_{uv},p^{T'}_{uv}]$, then set $T':=\sigma(T)$.
\end{enumerate}

This process clearly satisfies that with probability one, either $T'= T$ or $T'=\sigma(T)$. By the work above, this means that $r_S(T)=r_S(T')$ as desired, as long as we confirm that this process validly samples both $T$ and $T'$. It is easy to see that the process is valid for $T$: the match outcomes are clearly independent, and any team $w$ beats $x$ if and only if a uniformly random draw from $[0,1]$ is $< p^T_{wx}$ (which happens with probability $p^T_{wx}$, as desired). 

To see that the process is valid for $T'$, observe first that $u$ beats $v$ with probability exactly $p_{uv}^{T'}$, because $u$ beats $v$ whenever $q_{uv} < p^{T'}_{uv}$. Moreover, after conditioning on $q_{uv}$, the outcome of each $(u,w)$ match is either set according to $q_{w1}$ (an independent, uniform draw from $[0,1]$), or $q_{w2}$ (also an independent, uniform draw from $[0,1]$). In either case, these match results are all set independently and with the correct probability (because $T, T'$ are $S$-adjacent, and because $q_{w1},q_{w2}$ are iid because $p^T_{uw} = p^T_{vw}$ for all $w$). This completes the proof.
\end{proof}

\begin{proof}[Proof of Theorem~\ref{thm:RSEBeasy}]
We again simply let $c = 1$ and take the trivial bound on the gain in bad terminal events.

To bound the probability of a bad terminal event, observe that a bad terminal event occurs only when $u$ plays $v$ this round, which happens with probability exactly $1/(n'-1)$.\footnote{Recall that in RSEB, $n':=2^{\lceil \log_2(n)\rceil}$.} So RSEB satisfies $b = 1/(n'-1)$.

For the good terminal events, we claim that $\Pr[G] \geq 2/(n'-1)$. First, observe that once we show this, we can plug into Theorem~\ref{thm:recursive} and conclude $\alpha_2^{\text{RSEB}}(\mathcal{T}^{1/2}) \leq \frac{1\cdot 1/(n'-1)}{1/(n'-1)+2/(n'-1)} = 1/3$. 

To establish this bound, let $\ell_u$ denote the number of teams which beat $u$ but not $v$, $\ell_v$ denote the number of teams which beat $v$ but not $u$. Without loss of generality let $\ell_u \geq \ell_v$.

If $\ell_u+\ell_v= 0$, then Lemma~\ref{lem:anon} already establishes that $u$ and $v$ gain nothing by manipulating. If $\ell_u\geq 2$, then whenever $u$ plays a team which beats them, we have a good terminal event (because $u$ is already eliminated before having ever played the $(u,v)$ match, so manipulating the match has no impact). This happens with probability at least $2/(n'-1)$. 

If $\ell_u = 1, \ell_v = 0$, then let $w$ denote the unique team which beats $u$ but not $v$. We now claim that if $w$ plays \emph{either} $u$ or $v$ that we are in a good terminal event. Indeed, if $w$ plays $u$, then $u$ is eliminated having never played the $(u,v)$ match. If instead $w$ plays $v$, then $w$ is eliminated, but now there are no remaining teams which beat $u$ but not $v$ (or vice versa) and Lemma~\ref{lem:anon} asserts that there are no further gains from manipulation. The probability that $w$ plays $u$ or $v$ is $2/(n'-1)$, as desired.

Finally, if $\ell_u = \ell_v = 1$, let $w$ denote the unique team which beats $u$ but not $v$, and $x$ denote the unique team which beats $v$ but not $u$ (this case requires slightly more calculations than RDM). If $u$ plays $w$, \emph{or} $v$ plays $x$, then at least one of $u,v$ is eliminated before the $(u,v)$ match is played, and therefore no gains are possible. This happens with probability $2/(n'-1) - \frac{1}{(n'-1)(n'-3)}$. In addition, if $u$ plays $x$ \emph{and} $v$ plays $w$, then both $x$ and $w$ will be eliminated, and Lemma~\ref{lem:anon} asserts that there are no further gains from manipulation. This occurs with probability $\frac{1}{(n'-1)(n'-3)}$. Therefore, a good terminal event happens with probability at least $2/(n'-1)$, as desired. This handles all possible cases, and establishes that $g \geq 2/(n'-1)$ in all cases. Plugging into Theorem~\ref{thm:recursive} as described above completes the proof.
\end{proof}

\subsection{Omitted Proofs from Section~\ref{sec:epsilon}}\label{app:epsilon}

\begin{proof}[Proof of Lemma~\ref{lem:lb}] Consider a 3-player tournament $T$ with teams $u$, $v$, and $w$, $p^T_{uv} = p^T_{vw} = p^T_{wu} = \tfrac{1}{2} + \varepsilon$.\footnote{That is, if $\varepsilon = 1/2$, this is the same 3-cycle example from~\cite{SchneiderSW17}.}

Noting that $T$ is randomized, there are two possible types of deterministic outcomes: outcomes where there is a Condorcet winner, and outcomes which form a cycle (either $v$ beats $u$, $u$ beats $w$, and $w$ beats $v$, or vice versa). 

In $T$, each of the three players has probability $(\tfrac{1}{2} + \varepsilon)(\tfrac{1}{2} - \varepsilon)$ of being a Condorcet winner.

Call the cycle where $v$ beats $u$ ``cycle 1''; this occurs with probability $(\tfrac{1}{2} + \varepsilon)^3$. Call the opposing cycle (where $u$ beats $v$) ``cycle 2''; this occurs with probability $(\tfrac{1}{2} - \varepsilon)^3$.

Let $r$ be any Condorcet-consistent tournament rule. Denote by $\gamma_x$ the probability that the rule selects $x$ as the winner when cycle 1 occurs (for any $x \in \{u,v,w\}$). Denote by $\beta_x$ the probability that the rule selects $x$ as the winner when cycle 2 occurs. Recall that $r$ must select $x$ as the winner with probability $1$ when $x$ is the Condorcet winner.

Suppose that $u$ and $v$ collude so that $u$ throws their match to $v$. Specifically, let $T^{uv}$ denote the $\{u,v\}$-adjacent tournament to $T$ where $p^{T^{uv}}_{uv}=0$ (instead of $1/2+\varepsilon$). In $T^{uv}$, $v$ is a Condorcet winner with probability $\tfrac{1}{2} + \varepsilon$ (they need only beat $w$), cycle 2 occurs with probability $(\tfrac{1}{2} - \varepsilon)^2$ ($v$ must lose to $w$, who must lose to $u$), and $w$ is a Condorcet winner with probability $(\tfrac{1}{2} + \varepsilon)(\tfrac{1}{2} - \varepsilon)$. No other outcomes are possible. We then have:

\begin{align*}
r_u(T^{uv}) + r_v(T^{uv}) - r_u(T) - r_v(T) &= (\tfrac{1}{2}+\varepsilon) + (\beta_u + \beta_v)\cdot (\tfrac{1}{2} - \varepsilon)^2 \\
&\qquad- 2(\tfrac{1}{2} + \varepsilon)(\tfrac{1}{2} - \varepsilon) -(\beta_u + \beta_v)\cdot (\tfrac{1}{2} - \varepsilon)^3 -  (\gamma_u + \gamma_v)(\tfrac{1}{2} + \varepsilon)^3\\
&=2\varepsilon \cdot (\tfrac{1}{2}+\varepsilon) + (\beta_u + \beta_v)\cdot ((\tfrac{1}{2} - \varepsilon)^2- (\tfrac{1}{2} - \varepsilon)^3) -  (\gamma_u + \gamma_v)(\tfrac{1}{2} + \varepsilon)^3
\end{align*}
Identical calculations hold for $T^{vw}, T^{wu}$. We can sum these three together to achieve:
\begin{align*}
&(r_u(T^{uv}) + r_v(T^{uv}) - r_u(T) - r_v(T))\\
+&(r_v(T^{vw}) + r_w(T^{vw} )- r_v(T) - r_w(T))\\
+ &(r_w(T^{wu}) + r_u(T^{wu}) - r_w(T) - r_u(T))\\
 &\qquad\qquad=6\varepsilon \cdot (\tfrac{1}{2}+\varepsilon) + 2(\beta_u + \beta_v+\beta_w)\cdot ((\tfrac{1}{2} - \varepsilon)^2- (\tfrac{1}{2} - \varepsilon)^3) -  2(\gamma_u + \gamma_v+\gamma_w)(\tfrac{1}{2} + \varepsilon)^3\\
 &\qquad\qquad=6\varepsilon \cdot (\tfrac{1}{2}+\varepsilon) + 2\cdot (\tfrac{1}{2} - \varepsilon)^2\cdot (\tfrac{1}{2}+\varepsilon) -  2(\tfrac{1}{2} + \varepsilon)^3\\
 &\qquad\qquad=6\varepsilon \cdot (\tfrac{1}{2}+\varepsilon) + 2\cdot (\tfrac{1}{2}+\varepsilon) \cdot( (\tfrac{1}{2} - \varepsilon)^2-  (\tfrac{1}{2} + \varepsilon)^2)\\
 &\qquad\qquad=6\varepsilon \cdot (\tfrac{1}{2}+\varepsilon) + 2\cdot (\tfrac{1}{2}+\varepsilon) \cdot(-2\varepsilon)\\
 &\qquad\qquad=2\varepsilon \cdot (\tfrac{1}{2}+\varepsilon) 
\end{align*}

This means that at least one of the three coalitions can gain at least $\tfrac{2}{3}\varepsilon(\tfrac{1}{2}+\varepsilon) = \tfrac{1}{3} \varepsilon + \tfrac{2}{3} \varepsilon^2$.
\end{proof}

\begin{proof}[Proof of Lemma~\ref{lem:calcs}]
First, observe that $(\tfrac{1}{2}+\varepsilon)^i \cdot (\tfrac{1}{2}-\varepsilon)^{n-i} - (\tfrac{1}{2}+\varepsilon)^j \cdot (\tfrac{1}{2}-\varepsilon)^{n-j} \leq (\tfrac{1}{2}+\varepsilon)^n - (\tfrac{1}{2}-\varepsilon)^n$. Next, we claim that for all $n$, $(\tfrac{1}{2}+\varepsilon)^n - (\tfrac{1}{2}-\varepsilon)^n \leq 2\varepsilon$. We begin by showing that $(\tfrac{1}{2}+\varepsilon)^n - (\tfrac{1}{2}-\varepsilon)^n \leq (\tfrac{1}{2}+\varepsilon)^{n-1} - (\tfrac{1}{2}-\varepsilon)^{n-1}$ whenever $n \geq 2$. Indeed:

\begin{align*}
(\tfrac{1}{2}+\varepsilon)^n - (\tfrac{1}{2}+\varepsilon)^{n-1} &=-(\tfrac{1}{2} - \varepsilon)(\tfrac{1}{2} + \varepsilon)^{n-1}\\
(\tfrac{1}{2} - \varepsilon)^{n} - (\tfrac{1}{2} - \varepsilon)^{n-1} &= -(\tfrac{1}{2} + \varepsilon)(\tfrac{1}{2} - \varepsilon)^{n-1}.
\end{align*}

So when $n\geq 2$:
\begin{align*}
(\tfrac{1}{2}+\varepsilon)^n - (\tfrac{1}{2}+\varepsilon)^{n-1} \leq (\tfrac{1}{2} - \varepsilon)^{n} - (\tfrac{1}{2} - \varepsilon)^{n-1}\\
\Rightarrow (\tfrac{1}{2}+\varepsilon)^n - (\tfrac{1}{2}-\varepsilon)^{n} \leq (\tfrac{1}{2} + \varepsilon)^{n-1} - (\tfrac{1}{2} - \varepsilon)^{n-1}
\end{align*}

Now, applying this inductively yields:

$$(\tfrac{1}{2}+\varepsilon)^n - (\tfrac{1}{2}-\varepsilon)^{n} \leq (\tfrac{1}{2} + \varepsilon)^{1} - (\tfrac{1}{2} - \varepsilon)^{1} = 2\varepsilon.$$

This completes the proof whenever $n \geq 1$, as $(\tfrac{1}{2}+\varepsilon)^i \cdot (\tfrac{1}{2}-\varepsilon)^{n-i} - (\tfrac{1}{2}+\varepsilon)^j \cdot (\tfrac{1}{2}-\varepsilon)^{n-j} \leq (\tfrac{1}{2}+\varepsilon)^n - (\tfrac{1}{2}-\varepsilon)^n \leq 2\varepsilon$. When $n = 0$, the claim clearly holds.

\end{proof}

\end{document}